\documentclass[conference]{IEEEtran}
\IEEEoverridecommandlockouts

\usepackage{epsfig}
\usepackage{amsthm}
\usepackage{amssymb}
\usepackage{setspace}

\usepackage{psfrag}
\usepackage{subfigure}
\usepackage{url}
\usepackage{stfloats}
\usepackage{amsmath}
\usepackage{array}
\usepackage[version=3]{mhchem}
\usepackage{subfigure}
\usepackage{amssymb,amsmath}

\usepackage{amsfonts}

\usepackage{multirow}

\usepackage{dsfont}
\usepackage[latin1]{inputenc}
\usepackage{times}

\usepackage{multibib}

\newtheorem{rem}{Remark}

\def\mb{\mathbf}
\def\opn{\operatorname*}

\def\ss#1{{\sf #1}}

\def\vec#1{\mb{#1}}

\def\mat#1{\mb{\uppercase{#1}}}

\def\EspOp{\ss{E}}

\newcommand{\Esp}[2][5]{%
  \ifcase#1
     \EspOp\{ #2 \}
     \or \EspOp \bigl\{ #2 \bigr\}
     \or \EspOp \Bigl\{ #2 \Bigr\}
     \or \EspOp \biggl\{ #2 \biggr\}
     \or \EspOp \Biggl\{ #2 \Biggr\}
  \else
     \EspOp \left\{ #2  \right\}
\fi}

\newcommand{\Earg}[3][5]{%
  \ifcase#1
     \EspOp_{#3} \{ #2 \}
     \or \EspOp_{#3} \bigl\{ #2 \bigr\}
     \or \EspOp_{#3} \Bigl\{ #2 \Bigr\}
     \or \EspOp_{#3} \biggl\{ #2 \biggr\}
     \or \EspOp_{#3} \Biggl\{ #2 \Biggr\}
  \else
     \EspOp_{#3} \left\{ #2  \right\}
\fi}

\newcommand{\CEsp}[3][5]{%
  \ifcase#1
     \EspOp\{ #2 \mid #3 \}
     \or \EspOp \bigl\{ #2 \bigm\vert #3 \bigr\}
     \or \EspOp \Bigl\{ #2 \Bigm\vert #3 \Bigr\}
     \or \EspOp \biggl\{ #2 \biggm\vert #3 \biggr\}
     \or \EspOp \Biggl\{ #2 \Biggm\vert #3 \Biggr\}
  \else
     \EspOp \left\{ #2  \,\middle\vert\, #3 \right\}
\fi}

\def\Jacob{{\ss D}}

\def\log{\ss{log}}

\def\vecop{\ss{vec}}

\def\d{\opn{d}\!}

\begin{document}

\title{On MMSE Properties and I-MMSE Implications in Parallel MIMO Gaussian Channels}
\author{
\IEEEauthorblockN{Ronit Bustin \thanks{The work of R. Bustin,  M. Payar\'o and S. Shamai has been supported by
the European Commission in the framework of the FP7 Network of Excellence in Wireless Communications NEWCOM++. The work of R. Bustin and S. Shamai was also supported by the Israel Science Foundation.}}
\IEEEauthorblockA{\scriptsize{Dept. Electrical Engineering}\\
\scriptsize{Technion$-$IIT}\\
\scriptsize{Technion City, Haifa 32000} \\
\scriptsize{Israel}\\
\scriptsize{Email: bustin@tx.technion.ac.il}} \and
\IEEEauthorblockN{Miquel Payar\'o \thanks{The work of M. Payar\'o was also partially supported by the Spanish Government
under project TEC2008-06327-C03-03/TEC (FBMC-SILAN).} }
\IEEEauthorblockA{\scriptsize{Dept. of Radiocommunications}\\
\scriptsize{Centre Tecnol\`ogic de Telecomunicacions} \\
\scriptsize{de Catalunya (CTTC)}\\
\scriptsize{Castelldefels, Barcelona, Spain}\\
\scriptsize{Email: miquel.payaro@cttc.es}} \and
\IEEEauthorblockN{Daniel P.~Palomar \thanks{The work of D. P. Palomar has been supported in part by the Hong Kong Research Grant
DAG$\_$S08/09.EG05.}}
\IEEEauthorblockA{\scriptsize{Dept. of Elec. and Comp. Engineering}\\
\scriptsize{Hong Kong University of Science}\\
\scriptsize{and Technology}\\
\scriptsize{Clear Water Bay, Kowloon, Hong Kong}\\
\scriptsize{Email: palomar@ust.hk}} \and
\IEEEauthorblockN{Shlomo Shamai (Shitz)}
\IEEEauthorblockA{\scriptsize{Dept. Electrical Engineering}\\
\scriptsize{Technion$-$IIT}\\
\scriptsize{Technion City, Haifa 32000}\\
\scriptsize{Israel}\\
\scriptsize{Email: sshlomo@ee.technion.ac.il}} }

\maketitle


\begin{abstract}
This paper extends the ``single crossing point'' property of the scalar MMSE function, derived by Guo, Shamai and Verd\'u (first presented in ISIT 2008), to the parallel \emph{degraded} MIMO scenario. It is shown that the matrix $\boldsymbol{Q}(t)$, which is the difference between the MMSE assuming a Gaussian input and the MMSE assuming an arbitrary input, has, at most, a single crossing point for each of its eigenvalues. Together with the I-MMSE relationship, a fundamental connection between Information Theory and Estimation Theory, this new property is employed to derive results in Information Theory. As a simple application of this property we provide an alternative converse proof for the broadcast channel (BC) capacity region under covariance constraint in this specific setting.
\end{abstract}


\section{Introduction} \label{sec:introduction}
A fundamental relationship between estimation theory and information
theory for Gaussian channels was presented in \cite{IMMSE}; in
particular, it was shown that for the MIMO standard Gaussian
channel,
\begin{eqnarray} \label{eq:standardMIMOGaussian}
\boldsymbol{Y} = \sqrt{\textrm{snr}}\boldsymbol{H}\boldsymbol{X} +
\boldsymbol{N}
\end{eqnarray}
where $\boldsymbol{N}$ is a standard Gaussian $n$-dimensional random vector and $\boldsymbol{H}$ is a fixed channel matrix known to 
the receiver, then regardless of the input distribution on $\boldsymbol{X}$, the mutual information and
the minimum mean-square error (MMSE) are related (assuming
real-valued inputs/outputs) by
\begin{multline}
\frac{d}{d\textrm{snr}}
I(\boldsymbol{X};\sqrt{\textrm{snr}}\boldsymbol{H}\boldsymbol{X} +
\boldsymbol{N}) = \nonumber \\
\frac{1}{2}\mathbb{E}\{\parallel\boldsymbol{H}\boldsymbol{X} -
\boldsymbol{H}\mathbb{E} \{ \boldsymbol{X} |
\sqrt{\textrm{snr}}\boldsymbol{H}\boldsymbol{X} + \boldsymbol{N} \}
\parallel ^2\} \textrm{.}
\end{multline}
Here $\mathbb{E}\{X|Y\}$ stands for the conditional mean of $X$
given $Y$. This fundamental relationship and its generalizations
\cite{IMMSE,Palomar}, referred to as the I-MMSE relationships, have
already been shown to be useful in several aspects of information
theory: providing insightful proofs for entropy power inequalities
\cite{EPI}, revealing the mercury/waterfilling optimal power
allocation over a set of parallel Gaussian channels \cite{Water} and recently generalizing this result to MIMO Gaussian channels in \cite{MIMOPowerAllocation},
tackling the weighted sum-MSE maximization in MIMO broadcast
channels \cite{MSERegion}, illuminating extrinsic information of
good codes \cite{EXIT}, and enabling a simple proof of the
monotonicity of the non-Gaussianness of independent random variables
\cite{NonGaussian}.
In \cite{PROP,PROP_full} and later in \cite{EkremUlukus} it has been shown that using this relationship one can provide insightful and simple
proofs for multi-user single antenna problems such as the BC, the secrecy capacity problem, and the multi-receiver secrecy capacity region. In \cite{BustinFirst} this approach has been extended to the MIMO Gaussian wiretap channel yielding a closed form expression for the secrecy capacity. In order to provide the converse proof of the BC capacity region in \cite{PROP,PROP_full},
the authors proved an inherent property of the MMSE, the ``single crossing point'' property: as a function of \emph{snr}, the MMSE of the Gaussian input distribution and the MMSE of an arbitrary input distribution intersect at most once.
This property is stronger than required in order to prove the BC capacity region, however it is an interesting property on its own.

Motivated by this approach, our goal is to examine the properties of the MMSE matrix in the MIMO scenario, and relate
these properties to the mutual information. As an initial model we have
chosen the following simplified parallel channel,
\begin{eqnarray} \label{eq:generalModel}
\boldsymbol{Y} = \boldsymbol{H}\boldsymbol{X} + \boldsymbol{N}
\end{eqnarray}
where $\boldsymbol{X}$, $\boldsymbol{Y}$ and $\boldsymbol{N}$ are
\emph{n}-dimensional random vectors, and $\boldsymbol{N}$ is standard
Gaussian. $\boldsymbol{H}$ is assumed
diagonal and positive semidefinite. Note that $\boldsymbol{X}$ is not necessarily composed
of independent components, in which case the ``single crossing
point'' property can be deduced from the ``single crossing point''
property of the individual components.

As pointed out earlier, in the scalar channel scenario we have seen that a ``single crossing
point'' between the MMSE of the Gaussian input and the MMSE of an
arbitrary input distribution exists as a function of \emph{snr}.
The current more complex scenario, in which we have diagonal channel matrices, raises two questions:
What scalar function of the MMSE matrix should we examine for an analogous property to the ``single crossing point''?
And, along what $n \times n$-dimensional path should we look (since in our parallel MIMO scenario there are multiple $n \times n$-dimensional paths between every two channel matrices)? For consistency, even before answering these two questions, we would like to emphasize that the path will be parameterized through the scalar parameter $t$, thus instead of $\boldsymbol{H}$ we will write $\boldsymbol{R}(t)$, to avoid confusion.

The paper is organized as follows: section \ref{sec:definitions} contains the most basic definitions used in this work. Section \ref{sec:preliminaries} details our choice of path and gives some preliminary results used to prove our primary result. Section \ref{sec:singleCrossing} contains our main result, which is an extension of the single crossing point property to the parallel \emph{degraded} MIMO scenario. Section \ref{sec:connectionMutualInformation} connects the result of the previous section to the mutual information using the I-MMSE relationship. Finally, section \ref{sec:application} demonstrates how we can use this property to provide an alternative converse proof for the BC capacity region under covariance constraint in the parallel \emph{degraded} MIMO setting.


\section{Definitions} \label{sec:definitions}
We now formally give the definition of the MMSE matrix:
\begin{eqnarray} \label{eq:defineE}
\boldsymbol{E}(t) & = & \mathbb{E}\{ (\boldsymbol{X} - \mathbb{E}\{
\boldsymbol{X} | \boldsymbol{R}(t) \boldsymbol{X} + \boldsymbol{N}
\})  \nonumber \\ 
&  & (\boldsymbol{X} - \mathbb{E}\{ \boldsymbol{X} |
\boldsymbol{R}(t) \boldsymbol{X} + \boldsymbol{N} \})^T \} 
\end{eqnarray}
where $\boldsymbol{R}(t)$ corresponds to the channel matrix $\boldsymbol{H}$. The parameter $t$ determines the channel matrix, thus the new variable $\boldsymbol{R}(t)$ comes to highlight the dependence of the channel on the parameter $t$.

In this work we use the following I-MMSE
relationship derived by
Palomar and Verd$\acute{\textrm{u}}$ in \cite{Palomar}:
\begin{eqnarray} \label{eq:palomarVerdu}
\nabla_{\boldsymbol{H}}I(\boldsymbol{X};
\boldsymbol{H}\boldsymbol{X}+\boldsymbol{N}) =
\boldsymbol{H}\boldsymbol{E}
\end{eqnarray}
where $\boldsymbol{H}$ is a fixed known channel matrix, and
$\boldsymbol{N}$ is a standard Gaussian additive noise. Rewriting
this relationship as an integration along a path $\boldsymbol{R}(t)$
from $t=0$ to $t'$ results with the following expression:
\begin{eqnarray} \label{eq:lineIntegral}
I( \boldsymbol{X}; \boldsymbol{Y}(t')) & = & I( \boldsymbol{X}; \boldsymbol{R}(t') \boldsymbol{X} +
\boldsymbol{N})  \\  
& = & \int_{t=0}^{t'} \textrm{tr} \left( \left(
\boldsymbol{R}(t) \boldsymbol{E}(t) \right)^T  \boldsymbol{R}'(t)
\right) dt \nonumber
\end{eqnarray}
where $\boldsymbol{R}'(t) = \frac{d \boldsymbol{R}(t)}{dt}$.

In this work we specifically examine the properties of the
difference between the MMSE resulting from an arbitrary input
distribution and a Gaussian input distribution (not necessarily having the same covariance matrix). 
As such, we
require the following definition:
\begin{eqnarray} \label{eq:matrixQ}
\boldsymbol{Q}(t) & = & \boldsymbol{E}_G(t) - \boldsymbol{E}(t)
\end{eqnarray}
where we have denoted $\boldsymbol{E}_G(t)$ the MMSE matrix assuming a general Gaussian input distribution.

\section{Preliminaries} \label{sec:preliminaries}
We begin this section by presenting our choice of path, that is, we provide an answer to the second question presented in the Introduction: Along what $n \times n$-dimensional path should we look? In the scalar ``single crossing point'' property \cite{PROP,PROP_full}, the
MMSE is given as a function of \emph{snr}, a non-negative value. The
change in MMSE is examined as \emph{snr} monotonically increases. When switching to the MIMO scenario, our choice was to mimic the properties of the scalar scenario, that is, we show that there exists a
non-negative, monotonically non-decreasing path between any two
diagonal matrices $\boldsymbol{H}_1$ and $\boldsymbol{H}_2$, such
that $\boldsymbol{0} \preceq \boldsymbol{H}_1 \preceq
\boldsymbol{H}_2$, as given in the following lemma.

\newtheorem{LemmaPath}{Lemma}
\begin{LemmaPath} \label{th:LemmaPath}
For any two diagonal matrices $\boldsymbol{H}_1$ and
$\boldsymbol{H}_2$, such that $\boldsymbol{0} \preceq
\boldsymbol{H}_1 \preceq \boldsymbol{H}_2$, there exists a non-negative, monotonically non-decreasing path
$\boldsymbol{R}(t)$ for all $t
\in [0, 1]$ such that the following holds:
\begin{eqnarray} \label{eq:thPath}
\boldsymbol{R}(t=0) & = & \boldsymbol{0} \nonumber \\
\boldsymbol{R}(t_1) & = & \boldsymbol{H}_1 \nonumber \\
\textrm{and} \quad \boldsymbol{R}(t_2 = 1) & = & \boldsymbol{H}_2
\end{eqnarray}
where $0 < t_1 < t_2 = 1$.
\end{LemmaPath}
\begin{proof} 
We need to define a function, $g_i(t)$, for each diagonal element $i$. It suffices to choose any non-negative function $f_i(t)$ such that the area from $0$ to $t_1$ will equal $[\boldsymbol{H}_1]_{ii}$ and the area from $t_1$ to $t_2$ will equal $[\boldsymbol{H}_2]_{ii} - [\boldsymbol{H}_1]_{ii}$. Given that, we can set the function to be $g_i(t) = \int_0^t f_i(t') dt'$. The entire path, $\boldsymbol{R}(t)$, will be given by:
\begin{eqnarray} \label{eq:path}
\boldsymbol{R}(t) = \textrm{diag}\{ g_1(t),...,g_n(t) \} \textrm{.}
\end{eqnarray}
As required, this path passes between the zero matrix at $t=0$,
$\boldsymbol{H}_1$ at $t_1$ and $\boldsymbol{H}_2$ at $t_2=1$. Since $f_i(t)$ are chosen non-negative for all $i$ we
have a non-negative and monotonically non-decreasing path for all $t \in [0, 1]$.
\end{proof}

We now turn to provide some preliminary results that will be shown central in the sequel.
\newtheorem{dlambdadX}[LemmaPath]{Lemma}
\begin{dlambdadX}[\mbox{\cite[Ch.~4, Sec.~11]{Magnus}}] \label{lem:dlambdadX}
Let $\lambda_i$ and $\vec{u}_i$ indicate the $i$-th eigenvalue (assumed of multiplicity 1) of the
matrix $\boldsymbol{Z}$ and its corresponding eigenvector, respectively. Then, it
follows that
\begin{gather}
\Jacob_{\boldsymbol{Z}} \lambda_i = \vec{u}_i^T \otimes \vec{u}_i^T 
\end{gather}
where $\Jacob$ is the Jacobian operator, whose definition can be found in \cite{Magnus}. 
\end{dlambdadX}

\newtheorem{dlambdadt}{Corollary}
\begin{dlambdadt} \label{cor:dlambdadt}
If the matrix $\boldsymbol{Z}$ depends on a real scalar parameter $\tau$, \emph{i.e.}, $\boldsymbol{Z} =
\boldsymbol{Z}(\tau)$, then, assuming $\lambda_i$ is an eigenvalue of multiplicity 1, applying the
chain rule, we get:
\begin{align}
 \frac{\d \lambda_i}{\d \tau} &= \Jacob_{\boldsymbol{Z}(\tau)} \lambda_i \Jacob_{\tau} \boldsymbol{Z}(\tau)
\\ &= \big(\vec{u}_i^T \otimes \vec{u}_i^T\big) \vecop \left(
\boldsymbol{Z}'(\tau) \right) \\ &= \vec{u}_i^T \boldsymbol{Z}'(\tau)  \vec{u}_i
\end{align}
where the last equality follows from \cite[Ch.~2, Th.~2.2]{Magnus}.
\end{dlambdadt}

\newtheorem{dDdt}[dlambdadt]{Corollary}
\begin{dDdt} \label{cor:dDdt}
If, given a $\tau = \tau_0$, the matrix $\boldsymbol{Z}(\tau_0)$ is diagonal, we can always
take $[\vec{u}_i]_j = \delta_{ij}$ and, thus, the result in Corollary
\ref{cor:dlambdadt} particularizes to
\begin{gather}
\left. \frac{\d \lambda_i(\boldsymbol{Z}(\tau))}{\d \tau} \right|_{\tau = \tau_0} = \left. \frac{\d
[\boldsymbol{Z}(\tau)]_{ii}}{\d \tau} \right|_{\tau = \tau_0} .
\end{gather}
\end{dDdt}

\begin{rem}
Observe that the results in Lemma \ref{lem:dlambdadX} and Corollary
\ref{cor:dlambdadt} are valid only for the case where the multiplicity of
$\lambda_i$ is equal to 1. However, as explained in \cite[Ch.~8,
Sec.~12]{Magnus} and formally stated in \cite[Ch.~8, Sec.~12,
Th.~13]{Magnus}, in our case, the result in Corollary \ref{cor:dDdt} can be applied
directly to the case where the multiplicity of $\lambda_i$ is greater
than 1.



\end{rem}

\section{Single Crossing Point for each Eigenvalue} \label{sec:singleCrossing}
Before stating our primary result we require a lower bound on the matrix $\boldsymbol{Q}(t)$, defined in (\ref{eq:matrixQ}), which is given in the following lemma.
\newtheorem{LemmaLowerBound}[LemmaPath]{Lemma}
\begin{LemmaLowerBound} \label{th:LemmaLowerBound}
The following lower bound holds:
\begin{eqnarray} \label{eq:LemmaLowerBound}
\boldsymbol{Q}'(t) \succeq 2 \left( \boldsymbol{E}(t)
\boldsymbol{B}(t) \boldsymbol{E}^T(t) - \boldsymbol{E}_G(t)
\boldsymbol{B}(t) \boldsymbol{E}_G^T(t) \right)
\end{eqnarray}
where $\boldsymbol{B}(t) =
\boldsymbol{R}(t) \boldsymbol{R}'(t)$ is a diagonal matrix. 
\end{LemmaLowerBound}
\begin{proof}
We first provide the derivative of the MMSE with respect to the parameter $t$.
Using the chain rule given in \cite[equations (65-66)]{Palomar2},
\begin{eqnarray} \label{eq:chainRule}
D_t \boldsymbol{E}_{ij}(t) &=& D_H \boldsymbol{E}_{ij}(t) D_t \boldsymbol{R}(t) \nonumber \\
&=& \textrm{tr}\left( \frac{\partial \boldsymbol{E}_{ij}(t)}{\partial
\boldsymbol{R}(t)}^T \boldsymbol{R}'(t) \textrm{.} 
\right)
\end{eqnarray}
Since $\boldsymbol{R}(t)$ is diagonal, the last expression can be
further simplified to
\begin{eqnarray} \label{eq:chainRuleDiagGeneralized}
D_t \boldsymbol{E}_{ij}(t) & = & \sum_{l} \frac{\partial
\boldsymbol{E}_{ij}(t)}{\partial \boldsymbol{R}_{ll}(t)}
\boldsymbol{R}'_{ll}(t) \textrm{.} 
\end{eqnarray}
Using the result (\cite[eq. (131)]{Palomar2}),
\begin{eqnarray} \label{eq:resultPalomar2_generalGaussian}
D_{\boldsymbol{R}_{ll}(t)} \boldsymbol{E}_{ij}(t) & = & - \mathbb{E}
\{ \boldsymbol{\phi}_X (\boldsymbol{Y})_{jl}\left[ \boldsymbol{\phi}_X(\boldsymbol{Y})
\boldsymbol{R}(t)^T\right]_{il} \nonumber \\
& & + \boldsymbol{\phi}_X (\boldsymbol{Y})_{il}\left[
\boldsymbol{\phi}_X(\boldsymbol{Y})
\boldsymbol{R}(t)^T\right]_{jl} \} \nonumber \\
& = & - \mathbb{E} \{ \boldsymbol{\phi}_X (\boldsymbol{Y})_{jl}
\boldsymbol{\phi}_X(\boldsymbol{Y})_{il} \boldsymbol{R}(t)_{ll} \nonumber \\
& & + \boldsymbol{\phi}_X
(\boldsymbol{Y})_{il} \boldsymbol{\phi}_X(\boldsymbol{Y})_{jl}
\boldsymbol{R}(t)_{ll} \} \nonumber \\
& = & -2 \boldsymbol{R}_{ll}(t) \mathbb{E} \{ \boldsymbol{\phi}_X
(\boldsymbol{Y})_{jl} \boldsymbol{\phi}_X (\boldsymbol{Y})_{il} \}
\end{eqnarray}
where
\begin{eqnarray} \label{eq:phi}
\boldsymbol{\phi}_X(\boldsymbol{y}) = \mathbb{E} \{ (\boldsymbol{X} -
\mathbb{E} \{ \boldsymbol{X}|\boldsymbol{y}\})(\boldsymbol{X} -
\mathbb{E} \{ \boldsymbol{X}|\boldsymbol{y}\})^T |\boldsymbol{y} \} \textrm{.}
\end{eqnarray}
Note that
$\boldsymbol{\phi}_X(\boldsymbol{y})$ depends on $t$ through $\boldsymbol{Y}(t) =
\boldsymbol{R}(t)\boldsymbol{X} + \boldsymbol{N}$. The second
equality in equation (\ref{eq:resultPalomar2_generalGaussian}) is due to the fact
that $\boldsymbol{R}(t)$ is diagonal. Thus, we can write the derivative of $\boldsymbol{E}_{ij}(t)$ as
\begin{eqnarray} \label{eq:derivativeEij}
D_t \boldsymbol{E}_{ij}(t) & = & -2 \sum_{l}  \boldsymbol{R}_{ll}(t) \mathbb{E} \{ \boldsymbol{\phi}_X (\boldsymbol{Y})_{jl} \boldsymbol{\phi}_X (\boldsymbol{Y})_{il} \} \boldsymbol{R}'_{ll}(t) \nonumber \\
& = & -2 \sum_{l}  \boldsymbol{R}_{ll}(t) \boldsymbol{R}'_{ll}(t)
\mathbb{E} \{ \boldsymbol{\phi}_X (\boldsymbol{Y})_{jl} \boldsymbol{\phi}_X
(\boldsymbol{Y})_{il} \} \nonumber \\
& = & -2 \sum_{l} \boldsymbol{B}_{ll}(t) \mathbb{E} \{ \boldsymbol{\phi}_X
(\boldsymbol{Y})_{jl} \boldsymbol{\phi}_X (\boldsymbol{Y})_{il} \}
\end{eqnarray}
recalling that $\boldsymbol{B}_{ll}(t) =
\boldsymbol{R}_{ll}(t) \boldsymbol{R}'_{ll}(t)$. We can put this
expression into a matrix form as follows:
\begin{eqnarray} \label{eq:derivative_matrixE}
D_t \boldsymbol{E}(t) = -2 \sum_{l} \boldsymbol{B}_{ll}(t)
\mathbb{E} \{ \boldsymbol{\phi}_X (\boldsymbol{Y})_l \boldsymbol{\phi}_X
(\boldsymbol{Y})_{l}^T \}
\end{eqnarray}
where $\boldsymbol{\phi}_X (\boldsymbol{Y})_{l}$ is the $l^{th}$ column of the
matrix $\boldsymbol{\phi}_X (\boldsymbol{Y})$. Using the fact that for a
Gaussian input distribution $\boldsymbol{\phi}_X(\boldsymbol{Y})$ does not depend
on $\boldsymbol{Y}$ and, thus, $\boldsymbol{\phi}_X(\boldsymbol{Y}) = \mathbb{E} \{
\boldsymbol{\phi}_X(\boldsymbol{Y})  \} =\boldsymbol{E}^{G}(t)$
\cite{Palomar2}, we can obtain the following lower bound on the derivative of
the matrix $\boldsymbol{Q}(t)$:
\begin{align} 
& \boldsymbol{Q}'(t) \nonumber \\
& = 2 \sum_{l} \boldsymbol{B}_{ll}(t) \left( \mathbb{E} \{ \boldsymbol{\phi}_X
(\boldsymbol{Y})_l \boldsymbol{\phi}_X (\boldsymbol{Y})_{l}^T \} - \boldsymbol{E}^G_l (\boldsymbol{E}^G_l)^T \right) \nonumber \\
& \succeq 2 \sum_{l} \boldsymbol{B}_{ll}(t) \left( \mathbb{E} \{
\boldsymbol{\phi}_X
(\boldsymbol{Y})_l \} \mathbb{E} \{ \boldsymbol{\phi}_X (\boldsymbol{Y})_{l} \}^T - \boldsymbol{E}^G_l (\boldsymbol{E}^G_l)^T \right) \nonumber \\
& = 2 \sum_{l} \boldsymbol{B}_{ll}(t) \left( \boldsymbol{E}_l \boldsymbol{E}^T_l  - \boldsymbol{E}^G_l (\boldsymbol{E}^G_l)^T \right) \nonumber \\
& = 2 \left( \boldsymbol{E}(t) \boldsymbol{B}(t)
\boldsymbol{E}^T(t) - \boldsymbol{E}_G(t) \boldsymbol{B}(t)
\boldsymbol{E}_G^T(t) \right) \nonumber
\end{align}
where the inequality is due to Jensen. 
\end{proof}

Let us fix $t_0 \geq 0$ and consider the generalized eigenvalue decomposition \cite{GeneralizedEig} on $( \boldsymbol{E}_G(t_0), \boldsymbol{E}(t_0) )$. \footnote{In the generalized eigenvalue decomposition we have considered that $\boldsymbol{E}_G \succ \boldsymbol{0}$. A sufficient condition for $\boldsymbol{E}_G \succ \boldsymbol{0}$ is that the covariance of the Gaussian input distribution is non-singular and that $\boldsymbol{R}(t)$ is non-singular.} Thus, there exists an invertible matrix $\boldsymbol{V}_0$ such that,
\begin{eqnarray} \label{eq:GED}
\boldsymbol{E}_G(t_0) & = & \boldsymbol{V}_0^T \boldsymbol{V}_0 \nonumber \\
\boldsymbol{E}(t_0) & = & \boldsymbol{V}_0^T \boldsymbol{\Sigma}_0 \boldsymbol{V}_0
\end{eqnarray}
where $\boldsymbol{\Sigma}_0$ is a positive semi-definite diagonal matrix.
Thus,
\begin{eqnarray} \label{eq:GED2}
\boldsymbol{Q}(t_0) = \boldsymbol{E}_G(t_0) - \boldsymbol{E}(t_0) =
\boldsymbol{V}_0^T \left( \boldsymbol{I} - \boldsymbol{\Sigma}_0
\right) \boldsymbol{V}_0
\end{eqnarray}
and the following matrix:
\begin{eqnarray} \label{eq:Qtilde}
\boldsymbol{\widetilde{Q}}(t_0) =  \boldsymbol{V}_0^{-T} \boldsymbol{Q}(t_0) \boldsymbol{V}_0^{-1}
\end{eqnarray}
is diagonal.
By defining $\boldsymbol{C}_0 = \boldsymbol{V}_0 \boldsymbol{B}(t_0) \boldsymbol{V}_0^T$, we can rewrite the lower bound attained in Lemma \ref{th:LemmaLowerBound} as follows:
\begin{eqnarray} \label{eq:lowerBoundRewritten}
\boldsymbol{Q}'(t) \succeq 2 \boldsymbol{V}_0^T \left( \boldsymbol{\Sigma}_0 \boldsymbol{C}_0 \boldsymbol{\Sigma}_0 - \boldsymbol{C}_0 \right) \boldsymbol{V}_0 \textrm{.}
\end{eqnarray}

Our main result is the following:
\newtheorem{TheoremSingleCrossingEig}{Theorem}
\begin{TheoremSingleCrossingEig} \label{th:TheoremSingleCrossingEig}
Each eigenvalue of $\boldsymbol{Q}(t)$ crosses
 the horizontal axis at most once.
\end{TheoremSingleCrossingEig}
\begin{proof}
Consider the new matrix function $\widetilde{\boldsymbol{Q}}(t) = \boldsymbol{V}_0^{-T}
\boldsymbol{Q}(t) \boldsymbol{V}_0^{-1}$, which, from Sylvester's law of inertia, has the
same number of positive, negative, and zero eigenvalues as $\boldsymbol{Q}(t)$.

Note that $\widetilde{\boldsymbol{Q}}(t_0) = \boldsymbol{I} -
\boldsymbol{\Sigma}_0$ is a diagonal matrix and, thus, we can apply Corollary
\ref{cor:dDdt} to obtain a lower bound on the derivative of the
eigenvalues of $\widetilde{\boldsymbol{Q}}(t)$ evaluated at $t_0$:
\begin{align}
\left. \frac{\d \lambda_i \big( \widetilde{\boldsymbol{Q}}(t) \big)}{\d t} \right|_{t
= t_0} &= \left. \frac{\d \big[ \widetilde{\boldsymbol{Q}}(t) \big]_{ii}}{\d t}
\right|_{t
= t_0} 
= \left[\boldsymbol{V}_0^{-T} \boldsymbol{Q}'(t_0) \boldsymbol{V}_0^{-1}\right]_{ii} \nonumber \\
& \geq 2 ( [\boldsymbol{\Sigma}_0]_{ii} [\boldsymbol{C}_0]_{ii} [\boldsymbol{\Sigma}_0]_{ii}
- [\boldsymbol{C}_0]_{ii} ) \label{eq:bounddlambda}
\end{align}
where in (\ref{eq:bounddlambda}) we applied the lower bound given in Lemma
\ref{th:LemmaLowerBound}. 

Now, let us particularize the bound obtained in (\ref{eq:bounddlambda}) to the
non-positive eigenvalues of $\widetilde{\boldsymbol{Q}}(t_0)$, i.e., those that fulfill
that $\lambda_i \big( \widetilde{\boldsymbol{Q}}(t) \big) \leq 0$ which implies that
$[\boldsymbol{\Sigma}_0]_{ii} \geq 1$, from which it follows that
\begin{eqnarray}
\left. \frac{\d
\lambda_i \big( \widetilde{\mat{Q}}(t) \big)}{\d t} \right|_{t
= t_0} \geq 2 ( [\boldsymbol{\Sigma}_0]_{ii} [\boldsymbol{C}_0]_{ii} [\boldsymbol{\Sigma}_0]_{ii}
- [\boldsymbol{C}_0]_{ii} )
\geq 0
\end{eqnarray}
where we have used the fact that $\boldsymbol{B}(t) \succeq \boldsymbol{0}$.

The last result implies that, in a sufficiently small neighborhood of $t_0$, the
non-positive eigenvalues of $\widetilde{\boldsymbol{Q}}(t)$ are non-decreasing 
functions of $t$. Consequently, from the continuity of the eigenvalues, the number of negative eigenvalues of $\widetilde{\boldsymbol{Q}}(t)$ cannot increase.


Now, taking into account that the number of positive,
zero, and negative eigenvalues is preserved under the transformation
$\widetilde{\boldsymbol{Q}}(t) \mapsto \boldsymbol{Q}(t)$ we will informally show that a zero eigenvalue of $\boldsymbol{Q}(t)$
cannot become negative (a complete formal proof is given in \cite{full}). This will prove that each eigenvalue of $\boldsymbol{Q}(t)$ crosses the horizontal axis at most once.
We will prove this result by contradiction. Let's assume that there is a zero eigenvalue of $\boldsymbol{Q}(t_0)$ that becomes negative for $t > t_0$. Since the number of negative eigenvalues cannot increase, there must be at least one negative eigenvalue at $t_0$ that increases to zero.
If we examine the sign of the eigenvalues at $t_0 + \Delta$ for a sufficiently small $\Delta$,
we know that the zero eigenvalue has to be negative, however the negative eigenvalue (for sufficiently small $\Delta$)
is also still negative. Thus, we will have an increase in the number of negative eigenvalues, contradicting the property of no
increase in negative eigenvalues of $\widetilde{\boldsymbol{Q}}(t)$.
This shows that a zero eigenvalue cannot become negative in $\boldsymbol{Q}(t_0)$.
\end{proof}

The following corollary is a simple consequence from Theorem
\ref{th:TheoremSingleCrossingEig}.

\newtheorem{CorollarySingleCrossingEig}{Corollary}
\begin{CorollarySingleCrossingEig} \label{th:CorollarySingleCrossingEig}
If, for a given $t'$, the function $\boldsymbol{Q}(t')$ fulfils that $\boldsymbol{Q}(t') \succeq \boldsymbol{0}$ then
for all $t \geq t'$ we also have that $\boldsymbol{Q}(t) \succeq \boldsymbol{0}$.
\end{CorollarySingleCrossingEig}

Note that, by restricting the input distributions to be \emph{i.i.d.},
the matrix $\boldsymbol{Q}(t)$ is a diagonal matrix for all $t$. Thus, the single crossing property of the eigenvalues simplifies to
a single crossing property of the diagonal values, as expected due to the scalar single crossing property \cite{PROP,PROP_full}. However, in the general case, where the input distribution is arbitrary, the multivariate unique crossing property does not follow directly from the scalar case.

\section{Connecting to the Mutual Information} \label{sec:connectionMutualInformation}
As in the scalar scenario, our goal is to use the ``single crossings'' in order to derive
results regarding the mutual information. According to the I-MMSE relationship (\ref{eq:lineIntegral}), we would like to examine the following function:
\begin{eqnarray} \label{eq:differenceF}
\textrm{tr} \{ \boldsymbol{B}(t)
\boldsymbol{E}_G(t) \} - \textrm{tr} \{\boldsymbol{B}(t) \boldsymbol{E}(t) \} = \textrm{tr} \{ \boldsymbol{B}(t)
\boldsymbol{Q}(t) \} \textrm{.}
\end{eqnarray}
Since the trace is the sum of the eigenvalues, we need the following lemma, that extends our results regarding the eigenvalues of
$\boldsymbol{Q}(t)$ to the eigenvalues of $\boldsymbol{B}(t) \boldsymbol{Q}(t)$ for a positive semi-definite diagonal matrix $\boldsymbol{B}(t)$.

\newtheorem{BQ}[LemmaPath]{Lemma}
\begin{BQ} \label{th:BQ}
Each eigenvalue of $\boldsymbol{B}(t)\boldsymbol{Q}(t)$ crosses the
horizontal axis at most once.
\end{BQ}
\begin{proof}
Using,
\begin{eqnarray} \label{eq:BQ}
\lambda_i \{ \boldsymbol{B}(t) \boldsymbol{Q}(t) \} = \lambda_i \{ \boldsymbol{B}^{\frac{1}{2}}(t) \boldsymbol{Q}(t) \boldsymbol{B}^{\frac{1}{2}}(t) \}
\end{eqnarray}
with the fact that $\boldsymbol{B}(t)$ is diagonal and positive semi-definite, we again have the eigenvalues of a congruent transformation.
Thus, the extension of the previous claim follows directly.
\end{proof}

In order to use our results regarding the function
$\boldsymbol{Q}(t)$ we need the following lemma:

\newtheorem{GaussianExistence}[LemmaPath]{Lemma}
\begin{GaussianExistence} \label{th:GaussianExistence}
For any $t'$, there exists a Gaussian input covariance matrix $\boldsymbol{C}_G$
such that the following holds:
\begin{enumerate}
\item $\boldsymbol{C}_G \preceq \boldsymbol{C}_X$
\item $I( \boldsymbol{X}; \boldsymbol{Y}(t')) = I( \boldsymbol{X}_G;
\boldsymbol{Y}_G(t'))$
\item $\boldsymbol{Q}(t') \succeq \boldsymbol{0}$
\end{enumerate}
where $\boldsymbol{Y}(t') = \boldsymbol{R}(t') \boldsymbol{X} +
\boldsymbol{N}$ and $\boldsymbol{Y}_G(t') = \boldsymbol{R}(t') \boldsymbol{X}_G +
\boldsymbol{N}$.
\end{GaussianExistence}
\begin{proof} Due to the space limitations we will give only a sketch of the proof. For full details see \cite{full}.
From the third requirement we have:
\begin{eqnarray} \label{eq:definingJ}
\boldsymbol{Q}(t') = \boldsymbol{E}_G(t') - \boldsymbol{E}(t')
\equiv \boldsymbol{J} \succeq \boldsymbol{0} \textrm{.}
\end{eqnarray}
Furthermore, we can define:
\begin{eqnarray} \label{eq:definingC}
\boldsymbol{C} \equiv \boldsymbol{E}_L(t')
- \boldsymbol{E}(t') \succeq \boldsymbol{0}
\end{eqnarray}
where $\boldsymbol{E}_L(t')$ is the error covariance matrix assuming an optimal \underline{linear} estimator.
When $\boldsymbol{J} = \boldsymbol{0}$ ($\boldsymbol{E}_G(t') =
\boldsymbol{E}(t')$) we have that $\boldsymbol{Q}(t') =
\boldsymbol{0}$. According to Theorem \ref{th:TheoremSingleCrossingEig}
we have that all eigenvalues are non-positive for
$t \leq t'$. Furthermore, due to Lemma \ref{th:BQ} we conclude that the eigenvalues of
$\boldsymbol{B}(t)\boldsymbol{Q}(t)$ are also non-positive for all $t \leq t'$ and we can conclude, using (\ref{eq:lineIntegral}), that $I(
\boldsymbol{X}_G; \boldsymbol{Y}_G(t')) \leq I(\boldsymbol{X};
\boldsymbol{Y}(t'))$. If $\boldsymbol{J} = \boldsymbol{C}$ we have
that $\boldsymbol{C}_G = \boldsymbol{C}_X$ in which case we have $I(
\boldsymbol{X}_G; \boldsymbol{Y}_G(t')) \geq I(\boldsymbol{X};
\boldsymbol{Y}(t'))$. In order to comply with the first requirement 
we need to require that
$\boldsymbol{J} \preceq \boldsymbol{C}$. Thus, requirements 1 and 3 can be written using $\boldsymbol{J}$ and
$\boldsymbol{C}$, defined in equations (\ref{eq:definingJ}) and (\ref{eq:definingC}) respectively, and we have the following:
\begin{eqnarray}
I(\boldsymbol{X}_G; \boldsymbol{Y}_G(t')) \Bigg|_{\boldsymbol{J} = \boldsymbol{0} } \leq I(\boldsymbol{X};
\boldsymbol{Y}(t')) \leq I(\boldsymbol{X}_G; \boldsymbol{Y}_G(t')) \Bigg|_{\boldsymbol{J} = \boldsymbol{C} } \textrm{.} \nonumber 
\end{eqnarray}
The question is whether there exists such a
$\boldsymbol{J}$ that will also attain $I( \boldsymbol{X}_G;
\boldsymbol{Y}_G(t')) = I(\boldsymbol{X}; \boldsymbol{Y}(t')) \equiv
\alpha$. Both upper and lower bound can be expressed using the function $r(t) = \frac{1}{2} \log
\frac{|\boldsymbol{A}|}{|  \boldsymbol{B} + \boldsymbol{\Delta} \nu
|}$ which is continuous and monotonically decreasing
in $\nu$ for $0 \leq \nu \leq 1$ for $\boldsymbol{A} \succ \boldsymbol{0}$,
$\boldsymbol{B} \succ \boldsymbol{0}$ and $\boldsymbol{\Delta}
\succeq \boldsymbol{0}$ \cite{EkremUlukusMIMO}. Thus, there exists a $0 \leq \nu^* \leq 1$ that attains equality. That is, $\boldsymbol{J}^* = (1 - \nu^*) \boldsymbol{C}$
attains all three requirement.
\end{proof}



\section{Application: parallel degraded MIMO Gaussian BC} \label{sec:application}
As an example for the usage of these results we examine the parallel \emph{degraded} Gaussian BC.
We first note that the results attained so far have also been extended to the conditioned case where $(\boldsymbol{X}, U)$ are jointly distributed and $U - \boldsymbol{X} - \boldsymbol{Y}$ forms a Markov chain, but are omitted here due to space limitations. The conditioned MMSE is defined as:
\begin{eqnarray}
\label{eq:defineEu} \boldsymbol{E}^u(t) & = & \mathbb{E}\{
(\boldsymbol{X} - \mathbb{E}\{ \boldsymbol{X} | \boldsymbol{R}(t)
\boldsymbol{X} + \boldsymbol{N}, U \})  \nonumber \\ 
& &  (\boldsymbol{X} -
\mathbb{E}\{ \boldsymbol{X} | \boldsymbol{R}(t) \boldsymbol{X} +
\boldsymbol{N}, U \})^T \}
\end{eqnarray}
and the conditioned matrix $\boldsymbol{Q}^u(t) = \boldsymbol{E}^G(t) - \boldsymbol{E}^u(t)$,
which is the difference between the MMSE matrix assuming a general Gaussian input distribution independent of $U$, and the conditioned MMSE matrix.

We consider the \emph{degraded} parallel Gaussian BC channel:
\begin{eqnarray} \label{eq:model3}
\boldsymbol{Y}_1[m] & = & \boldsymbol{H}_1 \boldsymbol{X}[m] + \boldsymbol{N}_1[m] \nonumber \\
\boldsymbol{Y}_2[m] & = & \boldsymbol{H}_2 \boldsymbol{X}[m] +
\boldsymbol{N}_2[m]
\end{eqnarray}
where $\boldsymbol{N}_1[m]$ and $\boldsymbol{N}_2[m]$ are standard
additive Gaussian noise vectors, $\boldsymbol{H}_1$ and
$\boldsymbol{H}_2$ are diagonal positive definite matrices such that
$\boldsymbol{H}_1 \preceq \boldsymbol{H}_2$.
The channel input satisfies the covariance constraint: $\mathbb{E} \{ \boldsymbol{X} \boldsymbol{X}^T \} \preceq \boldsymbol{S}$,
where $\boldsymbol{S}$ is some positive definite matrix.

One way of proving that the Gaussian input achieves the
capacity region is by using the single-letter expression 
\cite{Comments}:
\begin{eqnarray} \label{eq:degradedCapacityRegion}
R_1 & \leq & I(U;\boldsymbol{Y}_1) \nonumber \\
R_2 & \leq & I(\boldsymbol{X};\boldsymbol{Y}_2|U)
\end{eqnarray}
where $U$ is an auxiliary random variable over a certain alphabet
that satisfies the Markov relation $U - \boldsymbol{X} -
(\boldsymbol{Y}_1,\boldsymbol{Y}_2)$. This was done for the scalar
Gaussian BC in \cite{PROP,PROP_full}. We will try to follow similar steps for
the \emph{degraded} parallel Gaussian channel.

Assume a pair
$(\boldsymbol{X}, U)$ with covariance matrix $\boldsymbol{C}_X$.
Using the conditioned version of Lemma \ref{th:GaussianExistence} we know that there exists a Gaussian distribution, with covariance $\boldsymbol{B} \preceq \boldsymbol{S}$, with the following properties:
\begin{equation} \label{eq:MIMOBC}
I(\boldsymbol{X} ; \boldsymbol{H}_1 \boldsymbol{X} + \boldsymbol{N}
| U)  = I(\boldsymbol{X}_G ; \boldsymbol{H}_1 \boldsymbol{X}_G +
\boldsymbol{N})
\end{equation}
\begin{align}
& I(\boldsymbol{X}_G ; \boldsymbol{H}_2 \boldsymbol{X}_G +
\boldsymbol{N} ) - I(\boldsymbol{X} ; \boldsymbol{H}_2
\boldsymbol{X} + \boldsymbol{N} |U )  \nonumber
\\ & = \int_{t=0}^{t_2}
\textrm{tr}\{ \boldsymbol{B}(t) \boldsymbol{Q}^u(t) \} dt \nonumber
\\ \label{eq:MIMOBC_inside}
& = \int_{t=0}^{t_1} \textrm{tr}\{ \boldsymbol{B}(t)
\boldsymbol{Q}^u(t) \} dt + \int_{t_1}^{t_2} \textrm{tr}\{
\boldsymbol{B}(t)
\boldsymbol{Q}^u(t) \} dt  \\
& = 0 + \int_{t_1}^{t_2} \textrm{tr}\{ \boldsymbol{B}(t)
\boldsymbol{Q}^u(t) \} dt \geq 0
\end{align}
where (\ref{eq:MIMOBC_inside}) is due to (\ref{eq:MIMOBC}),
and the inequality is due to the fact that $\boldsymbol{Q}^u(t)
\succeq \boldsymbol{0}$ for all $t \geq t_1$ and Lemma \ref{th:BQ}. Thus, we have a Gaussian
distribution that complies with a covariance constraint and also,
\begin{eqnarray} \label{eq:conclusion}
I(\boldsymbol{X} ; \boldsymbol{H}_1 \boldsymbol{X} + \boldsymbol{N}
| U) = I(\boldsymbol{X}_G ; \boldsymbol{H}_1 \boldsymbol{X}_G +
\boldsymbol{N} ) \nonumber \\
I(\boldsymbol{X} ; \boldsymbol{H}_2 \boldsymbol{X} + \boldsymbol{N}
|U) \leq I(\boldsymbol{X}_G ; \boldsymbol{H}_2 \boldsymbol{X}_G +
\boldsymbol{N} )
\end{eqnarray}
assuming a parallel \emph{degraded} model, that is, $\boldsymbol{0}
\prec \boldsymbol{H}_1 \preceq \boldsymbol{H}_2$. Now, substituting the above into the region given in equation (\ref{eq:degradedCapacityRegion}) we obtain the following region:
\begin{eqnarray} \label{eq:capacityRegionBCCovarianceConstraint}
R_1 & \leq & I( U ; \boldsymbol{Y}_1) = I( \boldsymbol{X};
\boldsymbol{Y}_1 ) - I( \boldsymbol{X}; \boldsymbol{Y}_1 | U)
\nonumber \\
& \leq & \frac{1}{2} \log | \boldsymbol{I} + \boldsymbol{H}_1
\boldsymbol{S} \boldsymbol{H}_1^T | - \frac{1}{2} \log |
\boldsymbol{I} +
\boldsymbol{H}_1 \boldsymbol{B} \boldsymbol{H}_1^T |  \nonumber \\
& = &  \frac{1}{2} \log  \frac{| \boldsymbol{I} + \boldsymbol{H}_1 \boldsymbol{S} \boldsymbol{H}_1^T  |}{| \boldsymbol{I} + \boldsymbol{H}_1 \boldsymbol{B} \boldsymbol{H}_1^T |} \\
R_2 & \leq & I( \boldsymbol{X}; \boldsymbol{Y}_2 | U) \leq
\frac{1}{2} \log | \boldsymbol{I} + \boldsymbol{H}_2 \boldsymbol{B}
\boldsymbol{H}_2^T | \textrm{.}
\end{eqnarray}
This concludes the converse part of the proof. The achievability is well-known using Gaussian superposition coding.


\bibliographystyle{IEEEtran}   

\bibliography{IEEEabrv,bib}

\end{document}